\documentclass[11pt,reqno]{amsart}

\usepackage{amssymb,amsmath,amsthm,amscd,latexsym,amsfonts}
\usepackage{mathtools}
\usepackage[T1]{fontenc}
\usepackage{graphicx,epstopdf}
\usepackage{graphicx}
\usepackage{xcolor}
\usepackage{cite}
\usepackage{float}
\usepackage[a4paper,top=2.9cm, bottom=2.9cm, left=2.9cm, right=2.9cm]{geometry}

\newtheorem{thm}{Theorem}
\newtheorem{defn}{Definition}
\newtheorem{lemma}{Lemma}
\newtheorem{pro}{Proposition}
\newtheorem{rk}{Remark}
\newtheorem{cor}{Corollary}

\numberwithin{equation}{section} 

\newcommand{\bea}{\begin{eqnarray}}
	\newcommand{\eea}{\end{eqnarray}}

\def\b{\beta}

\begin{document}
	\title [Kittel's zipper model on Cayley trees]
	{Kittel's  molecular zipper  model on Cayley trees}
	
	\author {U.A. Rozikov}
			 
	\address{ U.Rozikov$^{a,b,c}$\begin{itemize}
			\item[$^a$] V.I.Romanovskiy Institute of Mathematics,  9, Universitet str., 100174, Tashkent, Uzbekistan;
			\item[$^b$] AKFA University, National Park Street, Barkamol MFY,
			Mirzo-Ulugbek district, Tashkent, Uzbekistan;
			\item[$^c$] National University of Uzbekistan,  4, Universitet str., 100174, Tashkent, Uzbekistan.
	\end{itemize}}
	\email{rozikovu@yandex.ru}
	
	\begin{abstract}
Kittel's 1D model represents a natural DNA with two strands as a (molecular) zipper, which may separated as the temperature is varied. We define multidimensional version of this model on a Cayley tree and study the set of Gibbs measures.   We reduce description of Gibbs measures to solving of a non-linear functional equation, with unknown functions (called boundary laws) defined on vertices of the Cayley tree. Each boundary law defines a Gibbs measure. We give general formula of free energy depending on the boundary law. Moreover, we find some concrete boundary laws and corresponding Gibbs measures. Explicit critical temperature for occurrence a phase transition (non-uniqueness of Gibbs measures) is obtained.

\end{abstract}
\maketitle



\section{Introduction}

One of classical results is that there cannot be phase transitions in 1D systems with short range interactions. As noted in \cite{CS}, this assertion is  often receives the name ``van Hove's theorem'' \cite{LM}, but van Hove proved it for a system of hard-core segments
of a diameter $d_0$, that interact only at distances smaller than $d_1$ and assuming that the interaction potential is a continuous,
bounded below function. This result was generalized by Ruelle
\cite{Ru}.

At the same  time some models (Kittel's model, Chui-Weeks's model, Dauxois-Peyrard's model)  were found with phase
 transitions in spite of their 1D character and the range of their interactions (see \cite{CS} for details). 

In this paper we are going to study Kittel's model on a Cayley tree of arbitrary order $k\geq 1$. Note that for $k=1$ this tree coincides with 1D integer lattice $\mathbb Z^1$.  

We are interested to the set of Gibbs measures (in particular to problems of phase transitions) for the Kittel's model on trees. 

1D Kittel's model is defined as follows (see \cite{CS}, \cite{K}). Consider a single-ended zipper of $N$ links that
can be opened only from one end. This zipper is model of a natural DNA (see \cite{book} for structure of a DNA) two strands  of which may separated as the temperature is varied (see Fig. \ref{zip}). 
\begin{figure}[h!]
	\includegraphics[width=0.7\textwidth]{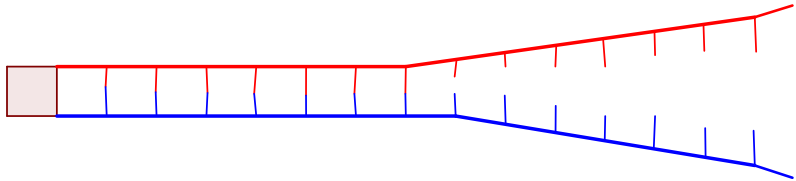}\\
	\caption{Single-ended zipper (DNA) with some open and closed links. }\label{zip}
\end{figure}

 If links $1, 2,..., n$ are all open, the energy
required to open link $n+1$ is $\epsilon$. However, if all the preceding links are not
open, the energy required to open link $n+1$ is infinite. 
The link $N$ (the end) cannot be
opened, and the zipper is said to be open when the first $N-1$ links are open. Suppose that there are $q$ orientations which each open link can
assume, i.e., the open state of a link is $q$-fold degenerated. 

 In \cite{CS} the 1D Kittel's model in solved in terms of a transfer matrix for the model's Hamiltonian defined as
\begin{equation}\label{H3}
H_N(s)=\epsilon\left(1-\delta_{0,s_1}\right)+\sum_{i=2}^{N-1}\left(\epsilon+J \delta_{0,s_{i-1}}\right)\left(1-\delta_{0,s_i}\right),
\end{equation}
where $s\in \{0,1,\dots,q\}^N$, $s_i=0$ means that link $i$ is closed, $s_i=1,2, \ldots, q$ means that the link is open in one of the possible $q$ states, and $\delta_{s, s^{\prime}}$ is the Kronecker symbol. 

 Kittel assumed that $J=\infty$, and the boundary condition $s_N=0$ (the leftmost end of the zipper is always closed). The partition function is given by
$$
Z_N=\sum_{s\in \{0,1,\dots,q\}^N:\atop s_N=0} \exp \left(-\beta H_N(s)\right),
$$
with $\beta=1 / T$ is the inverse temperature.

For the case $J=\infty$ by Kittel it was obtained that
$$Z_N={1-a^N\over 1-a}, \ \ \mbox{where} \ \ a=q\exp(-\beta\epsilon).$$

In order to have a phase transition, meaning a non-analyticity of the free energy, one can see that at $T_c={\epsilon \over \ln q}$, the derivative of the free energy is discontinuous. Moreover, note that $T_c$ is finite iff $q>1$; for the non-degenerate case $q=1$ (only one open state) there is no
phase transition.

In this paper we consider Kittel's model on the Cayley tree of arbitrary order $k\geq 1$. This model pretenses a set of interacting DNAs as introduced in \cite{Robp}, \cite{Rb}, \cite{Rm}. Note that the case $k=1$ coincides with 1D case, i.e., original Kittel's model. 

The paper is organized as follows.
In Section 2 we give main definitions. 
In Section 3 we reduce description of Gibbs measures of the Kittel's model to solving a non-linear functional equation, with unknown functions defined on vertices of the Cayley tree. Such a function is called a boundary law (see \cite[chapter 12]{Ge} and \cite[section 1.2.4]{BR}). Each boundary law defines a Gibbs measure. In Section 4, for each boundary law, we give formula of corresponding free energy. Sections 5 and 6 are devoted to finding concrete boundary laws (and corresponding Gibbs measures). We find explicitly critical temperature of occurrence of  phase transitions (non-uniqueness of Gibbs measures).

\section{Preliminaries}

For convenience of a reader let us recall some definitions (see  \cite{Robp}).\\

{\bf Cayley tree.} The Cayley tree $\Gamma^k$ of order $ k\geq 1 $ is an infinite tree,
i.e., a graph without cycles, such that exactly $k+1$ edges
originate from each vertex. Let $\Gamma^k=(V,L,i)$, where $V$ is the
set of vertices $\Gamma^k$, $L$ the set of edges and $i$ is the
incidence function setting each edge $l\in L$ into correspondence
with its endpoints $x, y \in V$. If $i (l) = \{ x, y \} $, then
the vertices $x$ and $y$ are called the {\it nearest neighbors},
denoted by $l = \langle x, y \rangle $. 

If an arbitrary edge $\langle x^0, x^1\rangle=l\in L$ is deleted from the Cayley tree $\Gamma^k$, it splits into two components, i.e., two identical semi-infinite trees $\Gamma^k_0$ and $\Gamma^k_1$ (see Fig. \ref{fT1}).
\begin{figure}[h!]
	\includegraphics[width=0.9\textwidth]{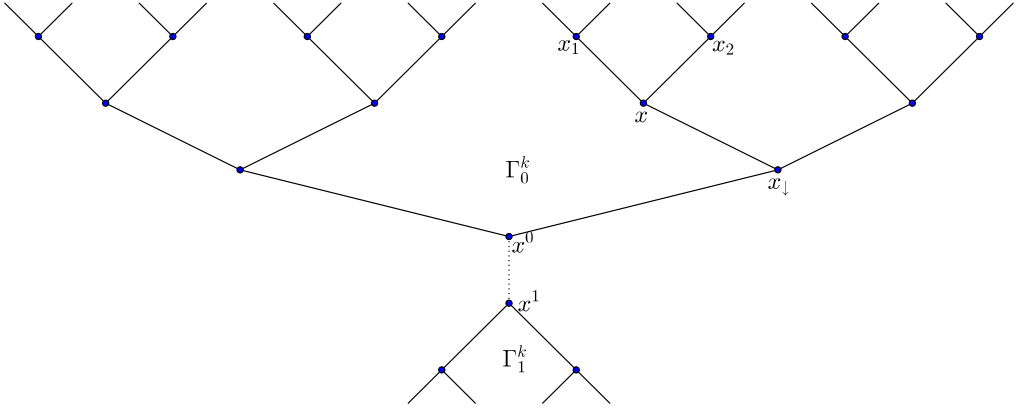}\\
	\caption{The Cayley tree of order $k=2$, separated to two semi-infinite sub-trees  $\Gamma^k_0$ and $\Gamma^k_1$. Shown an example of $x$, $x_{\downarrow}$ and $S(x)=\{x_1, x_2\}$. }\label{fT1}
\end{figure}
 In this paper we consider  semi-infinite Cayley tree $\Gamma^k_0=(V_0, L_0)$. The vertex $x^0$ is considered as a root of tree, the root has $k$ nearest neighbors and all other vertices of $\Gamma^k_0$ has $k+1$ nearest neighbors.

Denote by $S_1(x)$ the set of all nearest neighbors of $x\in V_0$. Let $x_{\downarrow}$ be unique vertex in $S_1(x)$ (where $x\in V_0\setminus\{x^0\}$) which is closer to $x^0$ than other elements of $S_1(x)$. Denote also 
$$S(x)=S_1(x)\setminus \{x_{\downarrow}\}.$$

The distance $d(x,y), x, y
\in V_0$ on the Cayley tree is the number of edges of the shortest
path from $x$ to $y$:
$$
d (x, y) = \min\{d \,|\, \exists x=x_0, x_1,\dots, x_{d-1},
x_d=y\in V_0 \ \ \mbox {such that} \ \ \langle x_0,
x_1\rangle,\dots, \langle x_{d-1}, x_d\rangle\} .$$

For the fixed $x^0\in V$ (the above defined root) and $n\geq 1$ we set 
$$|x|=d (x,
x^0), \ \ W_n = \ \{x\in V_0\ \ | \ \ |x| =n \}, $$
\begin{equation}\label{p*}
	V_n = \ \{x\in V_0\ \ | \ \ |x| \leq n \},\ \ L_n = \ \{l =
	\langle x, y\rangle \in L \  | \ x, y \in V_n \}.
\end{equation}

{\bf Configuration space.} Consider spin values from  $\Phi=\{0,1,\dots,q\}$, where $q\geq 1$. 

A configuration is any mapping $\sigma: x\in V \to \sigma(x)\in \Phi$. 
The vertex $x$ with $\sigma(x)=0$ means that the vertex is closed. Each $\sigma(x)=1,2,\dots q$ means that the vertex $x$ is open in one of possible  $q$ states.

For any $x\in V_0$ denote by $\pi_x$ the unique path connection $x^0$ and $x$. 

\begin{defn} A configuration $\sigma\in \Omega$ is called zipper-admissible if $\sigma(x^0)=0$ and from $\sigma(x)=0$ it follows that $\sigma(y)=0$ for all $y\in \pi_x$ (see Fig. \ref{fT2}).
\end{defn}	

In this paper we only consider zipper-admissible configurations.  Denote by $\Omega=\Phi^{V_0}$ the set of all such configurations. 

Admissible configurations in $V_n$ are defined analogously and the set of all such configurations is denoted by $\Omega_n$.

\begin{figure}[h!]
	\includegraphics[width=0.9\textwidth]{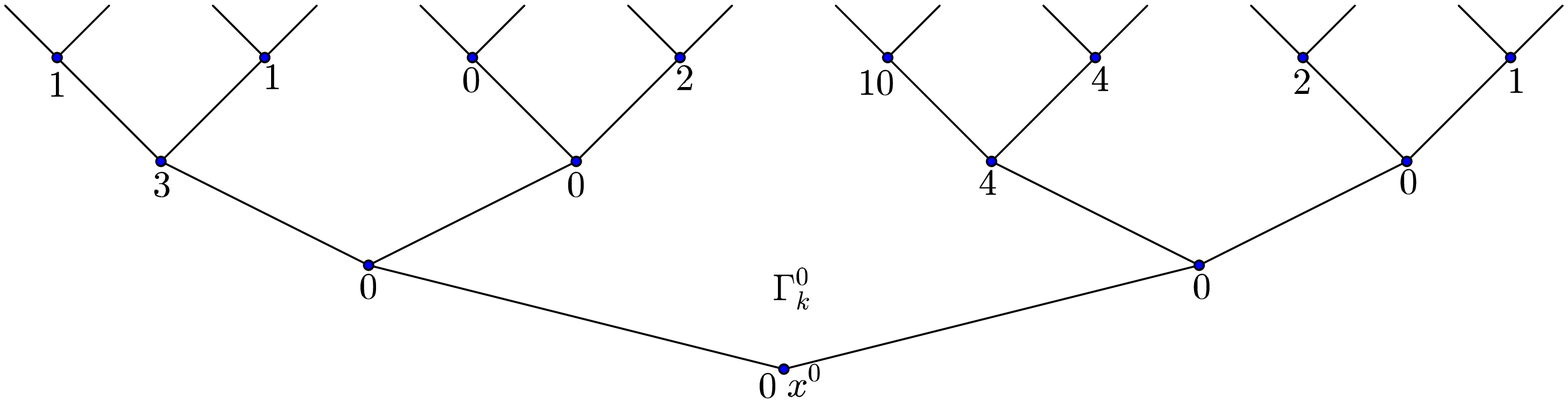}\\
	\caption{An example of zipper-admissible configuration on the Cayley tree of order $k=2$, with $q\geq 10$.}\label{fT2}
\end{figure}

{\bf The Hamiltonian of zipper model.} 
We consider the following model of the energy of the zipper-admissible  configuration $\sigma\in \Omega$ :
\begin{equation}\label{h}
	H(\sigma)=\epsilon \sum_{x\in W_1}(1-\delta_{0,\sigma(x)})+ \sum_{x\in V_0\setminus V_1}\left(\epsilon+J\delta_{0,\sigma(x_{\downarrow})}\right) \left(1-\delta_{0,\sigma(x)}\right),
\end{equation}
where $\epsilon\in\mathbb R$, $J\in \mathbb R\cup \{\infty\}$ are parameters, $\delta$ is the Kronecker delta.

\begin{rk} We note that
Hamiltonian (\ref{h}) generalizes (\ref{H3}) (firstly given in \cite{CS}) from $k=1$ to $k\geq 1$ and Kittel's model from $J=+\infty$ to $J\in \mathbb R\cup \{\infty\}$. 
\end{rk}

\section{System of functional equations of finite dimensional distributions}
Let $\Omega_n$ be the set of all zipper-admissible
configurations on $V_n$.

Define a finite-dimensional distribution of a probability measure $\mu$ on $\Omega_n$ as
\begin{equation}\label{*}
	\mu_n(\sigma_n)=Z_n^{-1}\exp\left\{-\beta H_n(\sigma_n)+\sum_{y\in W_n}h_{\sigma(y), y}\right\},
\end{equation}
where $\sigma_n\in \Omega_n$,  $\beta=1/T$, $T>0$ is temperature,  $Z_n^{-1}$ is the normalizing factor (partition function):
\begin{equation}\label{Zn}
	Z_n:=Z_n(\beta, h)=\sum_{\hat\sigma_n\in \Omega_n}\exp\left\{-\beta H_n(\hat \sigma_n)+\sum_{y\in W_n}h_{\hat\sigma(y), y}\right\},
\end{equation}
and
$$\{h_{a,  x}\in \mathbb R: \, a\in \Phi, \, x\in V_0\}$$ is a collection of real numbers and
$$	H_n(\sigma_n)=\epsilon \sum_{x\in W_1}(1-\delta_{0,\sigma(x)})+ \sum_{x\in V_n\setminus V_1}\left(\epsilon+J\delta_{0,\sigma(x_{\downarrow})}\right) \left(1-\delta_{0,\sigma(x)}\right).$$

We say that the probability distributions (\ref{*}) are compatible if for all
$n\geq 1$ and $s_{n-1}\in \Omega_{n-1}$:
\begin{equation}\label{**}
	\sum_{\omega_n\in \Omega_{W_n}}\mu_n(s_{n-1}\vee \omega_n)=\mu_{n-1}(s_{n-1}).
\end{equation}
Here $s_{n-1}\vee \omega_n$ is the concatenation of the configurations.

The following theorem gives a criterion for compatibility of finite-dimensional distributions.

\begin{thm}\label{ei} Probability distributions
	$\mu_n(\sigma_n)$, $n=1,2,\ldots$, in
	(\ref{*}) are compatible iff for any $x\in V_0\setminus \{x^0\}$
	the following equation holds:
	\begin{equation}\label{***}
		z_x=\prod_{y\in S(x)}\left(\theta z_y+\eta\right),
	\end{equation}
where $\theta={1\over q}\exp(\beta\epsilon)$, $\eta=\exp(-\beta J)$, 
and	$z_x=\exp(h_{0,x}-h_{i,x})$ (independent on $i=1,\dots,q$).
\end{thm}
\begin{proof}  {\sl Necessity.}  Suppose that (\ref{**}) holds; we want to prove (\ref{***}). Substituting (\ref{*}) into
	(\ref{**}),   for any $n\geq 2$ and configurations $\sigma_{n-1}$: $x\in V_{n-1}\mapsto\sigma_{n-1}(x)\in
	\Phi $ we obtain 
	\begin{equation}\label{uu}\begin{array}{ll}
			\frac{Z_{n-1}}{Z_n}\sum_{\omega_n\in\Omega_{W_n}}
			\exp\left[\sum_{x\in W_{n-1}}\sum_{y\in S(x)}\left\{
			-\beta\left(\epsilon+J\delta_{0,\sigma_{n-1}(x)}\right) \left(1-\delta_{0,\omega_n(y)}\right)
			+
			h_{\omega_n(y),y}\right\}\right]=\\[2mm]
			\exp\left(\sum_{x\in
				W_{n-1}}h_{\sigma_{n-1}(x),x}\right),
		\end{array}
	\end{equation}
	where $\omega_n$: $x\in W_n\mapsto\omega_n(x)$.
	
	From (\ref{uu}) we get:
	$$
	{Z_{n-1}\over Z_n}\sum_{\omega_n\in \Omega_{W_n}}
	\prod_{x\in W_{n-1}}\prod_{y\in S(x)} \exp\left[
	-\beta\left(\epsilon+J\delta_{0,\sigma_{n-1}(x)}\right) \left(1-\delta_{0,\omega_n(y)}\right)
	+
	h_{\omega_n(y),y}\right]=$$
	\begin{equation}\label{fh}
		\prod_{x\in W_{n-1}} \exp\,(h_{\sigma_{n-1}(x),x}).
	\end{equation}
	Fix $x\in W_{n-1}$, $n\geq 2$ and consider two configurations $\sigma_{n-1}=\overline{\sigma}_{n-1}$ and
	$\sigma_{n-1}=\tilde{\sigma}_{n-1}$ on $W_{n-1}$ which coincide on $W_{n-1}\setminus \{x\}$,
	and rewrite the equality (\ref{fh}) for $\overline{\sigma}_{n-1}(x)=0$,  and $\tilde{\sigma}_{n-1}(x)=i$ (where $i=1,\dots, q$), then dividing first of them to the second one we get
	$$\exp\,(h_{0,x}-h_{i,x})=\prod_{y\in S(x)}\frac{\sum_{j\in\Phi}\exp\left[
		-\beta(\epsilon+J)\left(1-\delta_{0,j}\right)
		+
		h_{j,y}\right]}{\sum_{j\in\Phi\setminus \{0\}}\exp\left[
		-\beta\epsilon\left(1-\delta_{0,j}\right)
		+
		h_{j,y}\right]}$$
	\begin{equation}\label{zi}
		=\prod_{y\in S(x)}\frac{e^{h_{0,y}}+\sum_{j=1}^q e^{
		-\beta(\epsilon+J)
		+
		h_{j,y}}}{\sum_{j=1}^qe^{
		-\beta\epsilon	+
		h_{j,y}}}.
	\end{equation}
From this equality we note that 
	$z_{i,x}=\exp(h_{0,x}-h_{i,x})$ must be independent on $i=1,\dots, q$. Therefore we introduce $z_x\equiv z_{i,x}$. Then from (\ref{zi}) we get
	 (\ref{***}). 
	
	{\sl Sufficiency.} Suppose that (\ref{***}) holds. It is equivalent to
	the representations
	\begin{equation}\label{ru}
		\begin{array}{ll}
		\prod_{y\in S(x)}\sum_{j\in\Phi}\exp\left[
		-\beta (\epsilon+J)\left(1-\delta_{0,j}\right)
		+
		h_{j,y}\right]= a(x)\exp\,(h_{0,x}), \\[3mm]
		\prod_{y\in S(x)}\sum_{j\in\Phi\setminus\{0\}}\exp\left[
		-\beta \epsilon\left(1-\delta_{0,j}\right)
		+
		h_{j,y}\right]= a(x)\exp\,(h_{i,x}), i\in \Phi\setminus\{0\}
		\end{array}
	\end{equation} for some function $a(x)>0, x\in V_0\setminus\{x^0\}.$
	We have
	\begin{equation}\label{ru1}
		{\rm LHS \ \ of \ \  (\ref{**})}=\frac{1}{Z_n}\exp(-\beta H(\sigma_{n-1}))\times
	\end{equation}
	$$
	\prod_{x\in W_{n-1}} \prod_{y\in S(x)}\sum_{j\in\Phi}\exp\left[
	-\beta(\epsilon+J\delta_{0,\sigma_{n-1}(x)})\left(1-\delta_{0,j}\right)
	+
	h_{j,y}\right].$$
	
	Substituting (\ref{ru}) into (\ref{ru1}) and denoting $A_n(x)=\prod_{x \in W_{n-1}} a(x)$,
	we get
	\begin{equation}\label{ru2}
		{\rm RHS\ \  of\ \  (\ref{ru1}) }=
		\frac{A_{n-1}}{Z_n}\exp(-\beta H(\sigma_{n-1}))\prod_{x\in W_{n-1}}
		\exp(h_{\sigma_{n-1}(x),x}).\end{equation}
	
	Since $\mu^{(n)}$, $n \geq 2$ is a probability, we should have
	$$ \sum_{\sigma_{n-1}\in\Omega_{V_{n-1}}} \sum_{\omega_n\in \Omega_{W_n}}\mu^{(n)} (\sigma_{n-1}, \omega_n) = 1. $$
	
	Hence from (\ref{ru2})  we get $Z_{n-1}A_{n-1}=Z_n$, and (\ref{**}) holds.
		\end{proof}

By this theorem the problem of description of Gibbs measures is reduced to the problem of solving the functional equation (\ref{***}). 
It is difficult to describe all solutions of this equation. Below under some conditions on parameters of the model we find several  solutions.

\begin{rk} From the proof of Theorem \ref{ei} we know that $z_{x}=\exp(h_{0,x}-h_{i,x})$ must be independent on $i=1,\dots, q$, i.e., $h_{1,x}=\dots =h_{q,x}$. Moreover, for each $x\in V_0\setminus \{x^0\}$ adding a number $\alpha_x$ to both 
$h_{0,x}$ and $h_{i,x}$ does not change the value $h_{0,x}-h_{i,x}$. Thus without loss of generality we can assume that $h_{0,x}+h_{1,x}=0$ for any $x$. Indeed, if $h_{0,x}+h_{1,x}\ne 0$ we consider new functions $\hat h_{0,x}=h_{0,x}+\alpha_x$ and $\hat h_{i,x}=h_{1,x}+\alpha_x$ where $\alpha_x={-1\over 2}(h_{0,x}+h_{1,x})$. Then 
$$z_{x}=\exp(h_{0,x}-h_{i,x})=\exp(\hat h_{0,x}-\hat h_{i,x}),$$
with $\hat h_{0,x}+\hat h_{i,x}=0$.
\end{rk}

A solution $z_{x}=\exp(h_{0,x}-h_{1,x})$ satisfying (\ref{***}) and $h_{0,x}+ h_{1,x}=0$  will be in the sequel called \emph{compatible}.

\section{Free energy}

Recall that a reduced free energy is
$$f_n(\beta, h)= {1\over |V_n|}\cdot \ln Z_n(\beta, h),$$
where $Z_n(\beta, h)$ is defined in (\ref{Zn}).

It is known that the reduced free energy encodes the relevant information about the physical system; for example, its first
derivative with respect to $\beta$ yields the internal energy while its second derivative yields, up to a multiplicative prefactor, the specific heat. For a finite volume, this is a (real) analytic function in $\beta$. But some failures of analyticity can only occur in the infinite volume limit:
$$f(\beta, h)=\lim_{n \to \infty} {1\over |V_n|}\cdot \ln Z_n(\beta, h).$$

If the resulting limit of the reduced free energy is an analytic function of $\beta$ (equivalently,
temperature), then the model has no (finite-temperature) phase transition. Finite values of $\beta$, or equivalently $T$, where analyticity fails are finite temperature
critical points corresponding to phase transitions.

For models given on a Cayley tree the free energy also depends on boundary law, i.e., solutions of equation (\ref{***}) (see \cite{RR} and references therein). 

\begin{thm}
For compatible  solution of (\ref{***}),
the free energy is given by the
formula
\begin{equation}\label{fe2}
	f(\b, h)=\lim_{n\to\infty}\frac{1}{ |V_n|}
	\sum_{x\in V_{n}}b(x),
\end{equation}
where\footnote{Note that $b(x)$ depends from $\beta$ because $\theta$, $\eta$ and solution $z_x$ are functions of $\beta$.}
$$
b(x)
=\frac{1}{ 2}\ln\left({\theta z_x+\eta\over \theta^{2}z_x}\right).
$$
\end{thm}
\begin{proof} By definitions (\ref{H3}), (\ref{Zn})   we take  initial value 
		\begin{equation}\label{Z1}
	Z_1(\beta, h)=\sum_{\sigma\in \Omega_1}\prod_{x\in W_1}\exp\left\{-\beta \epsilon (1-\delta_{0,\sigma(x)})+h_{\sigma(x), x}\right\},
	\end{equation}
 and use the following recurrent 
	formula (see the last line in the proof of Theorem \ref{ei}):
\begin{equation}\label{re}
	Z_n(\beta, h)=A_{n-1}Z_{n-1}(\beta, h), \ \ n\geq 2
\end{equation}
where $A_n=\prod\limits_{x\in W_n}a(x)$ with  some function $a(x)>0, x\in V\setminus\{x^0\}$ satisfying
	\begin{equation}\label{b}
	\begin{array}{ll}
		\prod_{y\in S(x)}\sum_{j\in\Phi}\exp\left[
		-\beta (\epsilon+J)\left(1-\delta_{0,j}\right)
		+
		h_{j,y}\right]= a(x)\exp\,(h_{0,x}), \\[3mm]
		\prod_{y\in S(x)}\sum_{j\in\Phi\setminus\{0\}}\exp\left[
		-\beta \epsilon\left(1-\delta_{0,j}\right)
		+
		h_{j,y}\right]= a(x)\exp\,(h_{1,x}).
	\end{array}
\end{equation} 
Multiply these equalities and using   $h_{0,x}+ h_{1,x}=0$ (i.e., $z_{x}=\exp(-2h_{1,x})$)  we obtain
$$a(x)
=\theta^{-k}
\prod_{y\in S(x)}\left[z^{-1}_{y}(\theta z_y+\eta)\right]^{1/2}
=\exp\Big(\sum_{y\in S(x)}b(y)\Big)=\prod_{y\in S(x)}\exp(b(y)).$$
Inserting this formula into the recursive equation (\ref{re}) and  by iteration we get 
$$Z_n(\b,h)=\prod_{x\in V_{n-1}}a(x)=\prod_{x\in V_{n-1}}\prod_{y\in S(x)}\exp(b(y))=\prod_{x\in V_{n}}\exp(b(x)), \ \ n\geq 2, $$
which gives (\ref{fe2}).
\end{proof}

\section{Case $J=\infty$}
In this case $\eta=0$ and the functional equation 
(\ref{***}) is reduced to 
	\begin{equation}\label{eo}
	z_x=\theta^k\prod_{y\in S(x)} z_y.
\end{equation}

{\bf Sub-case:} $k=1$. In this case we can identify vertices of $\Gamma_0^1$
as elements of $\mathbb N_0=\{0, 1, 2\dots \}$. Therefore the functional equation is reduced to 
$$z_n=\theta z_{n+1} \ \ \Leftrightarrow \ \ z_{n+1}={z_n\over \theta}, \ \ z_1>0,  \, n\geq 1.$$
Consequently
$$z_{n}=\theta^{-(n-1)} z_1, \ \ n\geq 1.$$
This is general solution of (\ref{eo}) for the case $k=1$. 

{\bf Sub-case:} $k\geq 2$. In this case we give a family of solutions.  To do this we assume that a solution $z_x$ depends on $|x|$ but not on $x$ itself, i.e., such a solution is in the form $$0<z_x=a_n\equiv a_n(\theta, k), \ \ \mbox{where} \ \ n=|x|\in \{1, 2, \dots \}.$$ Since $a_n>0$ and $\theta>0$ with $\theta\ne 1$ we can introduce $\alpha_n=\log_\theta a_n$. 
For arbitrary $\alpha_1\in \mathbb R$ define 
\begin{equation}\label{alpha}
	\alpha_n={1\over k^n}\left(\alpha_1 -{k(k^n-1)\over k-1}\right), \ \ n\geq 2.
	\end{equation}
\begin{pro} Let $k\geq 2$ and $\alpha_n$ be as in (\ref{alpha}) then $z_x=\theta^{\alpha_{|x|}} $, $x\in V_0\setminus\{x^0\}$ is a solution to (\ref{eo}) for any $\alpha_1\in \mathbb R$.	
\end{pro}
\begin{proof} Let $z_x=\theta^{\alpha_{|x|}}$ be a solution of (\ref{eo}) we show that $\alpha_n$ is as in (\ref{alpha}). 	If $|x|=n$ then for any $y\in S(x)$ we have $|y|=|x|+1=n+1$. Therefore,  from (\ref{eo}) we get 
	$$\alpha_{|x|}=k\left(\alpha_{|x|+1}+1\right), \ \ \mbox{i.e.} \ \ \alpha_{n+1}={\alpha_n\over k}-1, \, n\geq 1.$$
	It is easy to see that this recurrent formula
	gives  (\ref{alpha}). Moreover, with this $\alpha_{|x|}$ the function $z_x=\theta^{\alpha_{|x|}} $ satisfies (\ref{eo}). 
\end{proof}
\begin{rk} We note that the functional equation (\ref{eo}) has constant solution:
	$$z_x=z^*=\theta^{-{k\over k-1}}, \ \ \forall x\in V_0\setminus\{x^0\}.$$ Moreover, since $\lim_{n\to \infty}\alpha_n=-{k\over k-1}$ we have $\lim_{|x|\to \infty}z_x=z^*$ independently on $\alpha_1\in \mathbb R$.
\end{rk}
\begin{rk}
For case $\eta=0$ the free energy is independent on solutions (see  (\ref{fe2})). Derivative of the free energy at $\theta=1$ (i.e. $T_c={\epsilon \over \ln q}$) is discontinuous.  This generalizes Kittel's result, for $k\geq 1$, meaning a phase transition. 
\end{rk}

By Theorem \ref{ei} to each above mentioned solution corresponds a Gibbs measure, which we denote by $\mu_{z_1}$ (where $z_1>0$ for $k=1$ and $z_1=\theta^{\alpha_1}$ for $k\geq 2$) and $\mu_{z^*}$. 

We note that in case $\theta=1$ all measures coincide (because the solutions coincide). Moreover, by formula (\ref{*}) one can see that $\mu_{z_1}=\mu_{z'_1}$ iff $z_1=z'_1$.

Summarizing we obtain
\begin{thm}\label{tf} For any $k\geq 1$ and $\theta>0$, $\theta\ne 1$ there are uncountable many Gibbs measures $\mu_{z^*}$, $\mu_{z_1}$, $z_1>0$.  
\end{thm}
\begin{rk} Theorem \ref{tf} gives another view on the phase transition phenomenon: at critical value $\theta=1$ (i.e. $T_c={\epsilon \over \ln q}$) there is unique Gibbs measure; but for $T\ne T_c$ there are uncountable Gibbs measures. 
Thus uniqueness only at one temperature, this is quit different phenomenon compared with classical models (on Cayley trees) such as Ising model (see \cite{Run} and references therein) and Potts model (see \cite{Robp}), where uniqueness of Gibbs measure holds for a continuous infinite region of temperature but there are  uncountable measures as soon as a phase transition occurs.      	
\end{rk}

\section{Constant solutions for $J<+\infty$} In this section, for $J\ne \infty$ we solve (\ref{***}) 
in class of constant functions, i.e., assume $z_x\equiv z$.
Then we get 
\begin{equation}\label{bd} 
	z=\left(\theta z+\eta \right)^k. 
\end{equation}	
By definitions of parameters we have $\theta>0$, $\eta>0$ and  $z>0$.  Full analysis of equation (\ref{bd}) given in the following lemma
\begin{lemma}\label{ts} For $k\geq 2$, let $N_k(\theta, \eta)$ be the number of positive solutions ($z>0$) of (\ref{bd}). It has the following form
	$$N_k(\theta, \eta)=\left\{\begin{array}{lll}
		0, \ \ \mbox{if} \ \ \eta>\eta_c\\[2mm]
		1, \ \ \mbox{if} \ \ \eta=\eta_c\\[2mm]
		2, \ \ \mbox{if} \ \ 0<\eta<\eta_c,
	\end{array}
\right.$$
where $\eta_c\equiv \eta_c(k,\theta):={k-1\over k\sqrt[k-1]{k\theta}}.$ 
	\end{lemma}
\begin{figure}[h!]
	\includegraphics[width=0.6\textwidth]{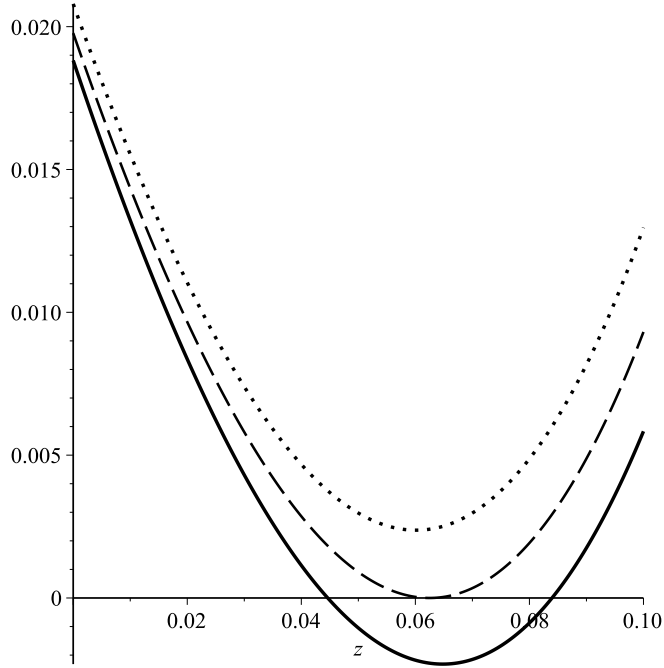}\\
	\caption{Plots of $f_4(z,2,\eta)$ for $\eta>\eta_c=3/8$ doted curve; for $\eta=\eta_c$ dashed curve; for $\eta<\eta_c$ solid curve.}\label{fe}
\end{figure}
\begin{proof}
Define 
$$f_k(z, \theta, \eta):=\left(\theta z+\eta \right)^k-z.$$
We have 
$$f_k(0, \theta, \eta)=\eta^k>0, \ \ f_k(+\infty, \theta, \eta)=+\infty,$$
$${\partial \over \partial z}f_k(z, \theta, \eta):=k\theta\left(\theta z+\eta \right)^{k-1}-1=0 \ \ \Leftrightarrow \ \ z=z^*:={1\over \theta}\left({1\over \sqrt[k-1]{k\theta}}-\eta\right).$$
Thus 
$f_k(z, \theta, \eta)$ is monotone decreasing if $0<z\leq z^*$ and increasing if $z>z^*$. Moreover,
\begin{equation}\label{zy}
	z^*>0 \ \ \Leftrightarrow \ \ \eta< {1\over \sqrt[k-1]{k\theta}}.
	\end{equation}
Under condition (\ref{zy}) we have
$$\min_{z\geq 0} f_k(z, \theta, \eta)=f_k(z^*, \theta, \eta)={1\over \theta}\left(\eta-{k-1\over k\sqrt[k-1]{k\theta}}\right).$$
Therefore equation (\ref{bd}) does not have a positive solution if $\min_{z\geq 0} f_k(z, \theta, \eta)>0$; it has unique positive solution if $\min_{z\geq 0} f_k(z, \theta, \eta)=0$ and 
it has exactly two solutions if $\min_{z\geq 0} f_k(z, \theta, \eta)<0$ (see Fig.\ref{fe}). This completes the proof.
\end{proof}
By this lemma, if $\eta < \eta_c$ (resp. $\eta = \eta_c$) then there are exactly two (resp. one ) solutions to (\ref{bd}),  denote them by $z_1$ and $z_2$ (resp. $z_1$). Let $\mu_i$ be translation invariant Gibbs measure (TIGM) corresponding (by Theorem \ref{ei}) to solution $z_x\equiv z_i$, $i=1,2$. 
From Lemma \ref{ts} by Theorem \ref{ei} we get the following
\begin{thm}\label{ud} Let $N_{TIGM}(k, \theta, \eta)$ denote the number of TIGMs of the Kittel's model (where $\eta>0$). Then  
		$$N_{TIGM}(k, \theta, \eta)=\left\{\begin{array}{lll}
		0, \ \ \mbox{if} \ \ \eta>\eta_c\\[2mm]
		1, \ \ \mbox{if} \ \ \eta=\eta_c\\[2mm]
		2, \ \ \mbox{if} \ \ \eta<\eta_c.
	\end{array}
	\right.$$
\end{thm}
Using $\theta={1\over q}\exp(\beta\epsilon)$, $\eta=\exp(-\beta J)$, $\beta=1/T$ from equation $\eta=\eta_c$, with respect to $T>0$, we get the following critical temperature 
\begin{equation}
	T_{\rm cr}\equiv T_{\rm cr}(k, q, \epsilon, J)={\epsilon-(k-1)J\over \ln\left(q\,{(k-1)^{k-1}\over k^k}\right)}.
\end{equation}
Consider the following set 
$$A=\{(k, q, \epsilon, J)\in \mathbb N^2\times \mathbb R^2\, : \, T_{\rm cr}(k, q, \epsilon, J)>0\}.$$
From Theorem \ref{ud} we get

\begin{cor}
	For the Kittel's model on the Cayley tree of order $k\geq 1$, if $(k, q, \epsilon, J)\in A$ then there is a phase transition. Moreover, non-uniqueness of Gibbs measure holds for $T<T_{\rm cr}$ (resp. $T>T_{\rm cr}$) if $J>0$ (resp.  $J<0$). 
\end{cor}

{\bf Example.} Consider the case $k=2$, i.e., Cayley tree of order two. In this case  two solutions mentioned in Lemma \ref{ts} can be found explicitly:
$$z_{\pm}:={1-2\theta \eta\pm\sqrt{1-4\theta\eta}\over 2\theta^2}.$$
Here $\eta<\eta_c={1\over 4\theta}$, i.e.,   $T_{\rm cr}= {\epsilon-J\over \ln(q/4)}$.
Corresponding free energies are
$$f(\beta, z_{\pm})= -{1\over 4}\ln(\theta^4z_{\pm}).$$

In Fig.\ref{free} we give graphs of the free energies. 
\begin{figure}[h!]
	\includegraphics[width=0.5\textwidth]{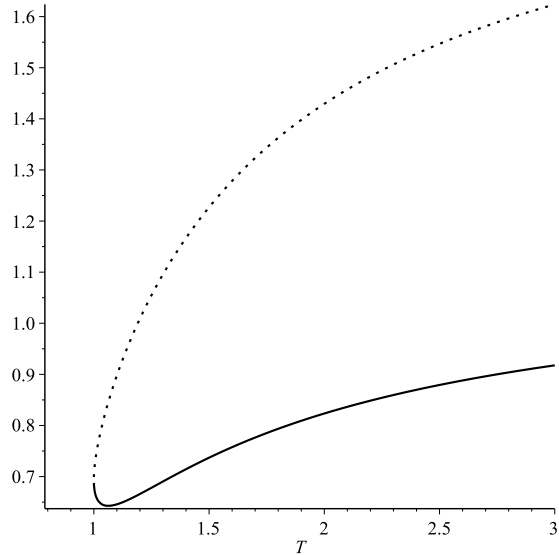}\\
	\caption{Graphs of free energy $f({1\over T}, z_{+})$ (the solid line) and $f({1\over T}, z_{-})$ (the doted line) on $1=T_{\rm cr}\leq T\leq 3$ in case $k=2$, $q=8$, $\epsilon=2\ln 2$, $J=\ln 2$.}\label{free}
\end{figure}

\section*{Acknowledgements}

The author thanks Institut des Hautes \'Etudes Scientifiques (IHES), Bures-sur-Yvette, France for support of his visit to IHES.  This visit was partially supported by a travel grant from the IMU-CDC.

\section*{Statements and Declarations}

{\bf	Conflict of interest statement:} 
The author states that there is no conflict of interest.

\section*{Data availability statements}
The datasets generated during and/or analyzed during the current study are available from the author on reasonable request.

\end{document}